\documentclass[preprintnumbers,amsmath,amssymb,superscriptaddress]{revtex4}
\usepackage{amsthm}
\usepackage{enumerate}

\newtheorem{theorem}{Theorem}[section]
 \newtheorem{lemma}[theorem]{Lemma}
 \newtheorem{corollary}[theorem]{Corollary}
 \newtheorem{definition}[theorem]{Definition}
 \newtheorem{proposition}[theorem]{Proposition}

\newtheorem{example}[theorem]{Example}

\newcommand*{\bbC}{\mathbb{C}}
\newcommand*{\complex}{\bbC}

\newcommand*{\cA}{\mathcal{A}}

\newcommand*{\cN}{\mathcal{N}}

\newcommand*{\cR}{\mathcal{R}}

\newcommand*{\cS}{\mathcal{S}}

\newcommand*{\id}{\openone}

\newcommand*{\tr}{\mathsf{tr}}
\newcommand*{\Sym}{\mathsf{Sym}}

\newcommand*{\ket}[1]{|#1\rangle}
\newcommand*{\bra}[1]{\langle #1|}
\newcommand*{\proj}[1]{\ket{#1}\bra{#1}}

\newcommand{\braket}[2]{\langle #1|#2\rangle}       
\newcommand*{\ot}{\otimes}
\newcommand*{\diag}{\mathsf{diag}}

\begin{document}

\title{The hierarchical structure of local unitary invariants.}
\author{Graeme \surname{Mitchison}} \email[]{gjm12@cam.ac.uk}
\affiliation{Centre for Quantum Information and Foundations, DAMTP,
University of Cambridge, Cambridge CB3 0WA, UK}

\begin{abstract}

  Local unitary invariants allow one to test whether multipartite
  states are equivalent up to local basis changes. Equivalently, they
  specify the geometry of the ``orbit space'' obtained by factoring
  out local unitary action from the state space. This space is of
  interest because of its intimate relationship to
  entanglement. Unfortunately, the dimension of the orbit space grows
  exponentially with the number of subsystems, and the number of
  invariants needed to characterise orbits grows at least as
  fast. This makes the study of entanglement via local unitary
  invariants seem very daunting. I point out here that there is a
  simplifying principle: Invariants fall into families related by the
  tracing-out of subsystems, and these families grow exponentially
  with the number of subsystems. In particular, in the case of pure
  qubit systems, there is a family whose size is about half the
  dimension of orbit space. These invariants are closely related to
  cumulants and to multipartite separability. Members of the family
  have been repeatedly discovered in the literature, but the fact that
  they are related to cumulants and constitute a family has apparently
  not been observed.
\end{abstract}

\maketitle

\section{Introduction}

Given a multipartite quantum state, the action of local unitaries maps
out an orbit of locally equivalent states. It would be very
interesting to understand the structure of the space of such
orbits. Polynomial invariants of the local unitary action provide a
way to do this. These are polynomials in the coefficients of a state
(and the complex conjugates of these coefficients) that are invariant
when the coefficients change under local unitary action.  It is known
that a finite set of invariants, a {\em fundamental set}, suffices to
generate all the polynomial invariants and to distinguish local
unitary orbits in state space \cite{Sturmfels,SudberyOpen,OV}. The
geometry of the orbit space is determined by the invariants in a
fundamental set and the algebraic relations between them. Indeed, the
ring of invariants can be viewed as the coordinate ring of the orbit
space, whose geometrical structure is therefore specified by the
algebraic relations amongst invariants.

Fundamental sets of invariants have been determined for pure states of
two and three qubits
\cite{LindenPopescu98,Sudbery01,Grassl98,Acinetal00,SudberyOpen}, and
for two-qubit mixed states \cite{Grassl98,KingWelsh}. Finding
fundamental sets is generally a hard problem. One reason for this is
that the dimension $D(n)$ of the orbit space grows exponentially with
the number, $n$, of systems, and the size of a fundamental set is
generally far larger. I point out here that there is a redeeming
feature: an invariant in $n$ systems gives rise to a set of invariants
in $n+1$ systems via a tracing-out operation. This means that any
invariant generates a family that grows exponentially with the number
of systems. So, given an invariant, one gets an exponential family of
them for free. I illustrate this by defining a family of
invariants for pure $n$-qubit states that grows exponentially with $n$
and asymptotically has $\frac{1}{2}D(n)$ members. Moreover, for any
$n$, the members of this family are algebraically independent, which
means that they are useful candidates for building a fundamental set.

The key to the construction of this family is the set of joint
cumulants
\cite{Thiele1,Thiele2,Fisher1,Fisher2,KendallStuart77,Royer}. Cumulants
are most commonly encountered as statistical tools, or as ingredients
in cluster expansions. For our purposes, they are simply polynomials
in the coefficients of states that have certain desirable properties;
for instance, they are closely related to the separability of
multiparty states. Cumulants can be introduced in an attractive if
unorthodox way by giving the space of vectors in
$(\complex^d)^{\otimes n}$ the structure of an algebra in which a
Taylor series and hence analystic functions can be defined. In
particular, one can define a log function with the property that
\[
\log (\ket{\psi} \otimes \ket{\phi})= \log (\ket{\psi})+ \log (\ket{\phi}),
\]
for all $\ket{\psi}$ and $\ket{\phi}$; see (\ref{tprod}) and
(\ref{linearise}) for a more precise statement. The coefficients in the
log-expansion are cumulants, which is how the connection between
cumulants and separability comes about.

This algebra is a rather unnatural construction, since it depends on a
particular choice of basis. However, twirling with respect to local
unitaries allows one to remove this artificiality and to generate a
set of invariants. It turns out that these cumulant-based invariants
are already known in the literature, though their relationship to
cumulants is apparently not recognised. They account for five of the
six independent 3-qubit invariants and eleven of the nineteen 4-qubit
invariants given in \cite{Luque07}. Asymptotically, there are $O(2^n)$
of them whereas $D(n)$ is $O(2^{n+1})$. There are also other families
of invariants, including a family of 4th degree polynomials that is of
size $O(2^{n-1})$. Breaking down the invariants into (generally
overlapping) hierarchical families gives a perspective on their
structure and on the entanglement of states.

The orbit space can be explored in other ways. Kraus
\cite{Kraus09,Kraus10} showed how to reduce multipartite qubit states
to a standard form that is invariant under local unitaries. An
alternative, geometric procedure that applies to local spaces of
arbitrary dimension (not just qubits) has also been proposed
\cite{Sawicki10}. These methods enable one to determine if two states
belong to the same orbit, and give insights that complement those
obtained from polynomial invariants. Finally, there is alternative way
of deriving invariants from cumulants (see Section
\ref{zhou-section}), due to Zhou et al \cite{Zhouetal06}. Despite
certain formal similarities, these seem not to have a simple
functional relationship to our invariants.

\section{The algebra of multi-partite states}\label{algebra}

Let $\cA_n^{(d)}$ be the the commutative algebra over $\complex$ with
generators $e_i$, $i=1, \ldots, n$ satisfying $e_i^d=0$. An element
$\psi$ of $\cA_n^{(d)}$ has the form
\[
\psi=\sum a_{i_1 \ldots i_n}e^{i_1}\ldots e^{i_n},
\]
where $0 \le i_k \le d-1$ and the $a$'s are complex coefficients. Then
$\psi=a+r$, where $a=a_{0 \ldots 0}$, and $r$ consists of at most
$d^n-1$ terms and satisfies $r^{n(d-1)+1}=0$. Any analytic function
$f:\cA_n^{(d)} \to \cA_n^{(d)}$ can be expanded in a Taylor series in r:
\begin{align}\label{taylor}
  f(\psi)=f(a)+f'(a)r+f''(a)r^2/2!+... 
\end{align}
This series is finite, because of the nilpotency of $r$, and is thus
well defined. For instance, in $\cA_n^{(2)}$ any $\psi$ can be written
\begin{align}
\psi=a_{00}+a_{10}e_1+a_{01}e_2+a_{11}e_1e_2
\end{align}
and if $a_{00} \ne 0$
\begin{align}
\log \psi&=\log (a_{00}+r)=\log a_{00}+\log(1+r/a_{00})=\log a_{00}+r/a_{00}-r^2/(2a_{00}^2), \nonumber\\
&=\log a_{00}+(a_{10}/a_{00})e_1+(a_{01}/a_{00})e_2+(a_{11}/a_{00}-a_{10}a_{01}/a_{00}^2)e_1e_2.\label{log}
\end{align}
Similarly, there is a finite polynomial for the algebra inverse; e.g.
\begin{align}
\psi^{-1}= a_{00}^{-1}-(a_{10}/a_{00}^2)e_1-(a_{01}/a_{00}^2)e_2-(a_{11}/a_{00}^2-2a_{10}a_{01}/a_{00}^3)e_1e_2;\label{inv}
\end{align}
and for other functions, such as $\exp$. These functions have all the
expected properties, e.g. 
\begin{align}
\psi \psi^{-1}&=1,\\
\label{logsum2} \log(\psi\phi)&=\log(\psi)+\log(\phi),\\ 
\label{expsum}\exp(\psi+\phi)&=\exp(\psi)\exp(\phi),\\
\exp(\log)(\psi)&=\psi.
\end{align}

Now identify $e_1^{i_1} \ldots e_n^{i_n}$ with the $n$-qudit basis
element $\ket{i_1 \ldots i_n}$. This sets up an isomorphism between
elements $\psi$ of $\cA_n^{(d)}$ and unnormalised $n$-qudit states
$\ket{\psi}$. For instance,
$\psi=a_{00}+a_{10}e_1+a_{01}e_2+a_{11}e_1e_2$ in $\cA_n^{(2)}$ can be
identified with the two qubit state
$\ket{\psi}=a_{00}\ket{00}+a_{10}\ket{10}+a_{01}\ket{01}+a_{11}\ket{11}$.
One can then carry across the structure of the algebra. For instance,
if
$\ket{\phi}=b_{00}\ket{00}+b_{10}\ket{10}+b_{01}\ket{01}+b_{11}\ket{11}$
we have the product
\begin{align*}
  \ket{\psi}\ket{\phi}&=a_{00}b_{00}\ket{00}+(a_{00}b_{10}+a_{10}b_{00})\ket{10}+(a_{00}b_{01}+a_{01}b_{00})\ket{01}\\
  &+(a_{00}b_{11}+a_{10}b_{01}+a_{01}b_{10}+a_{11}b_{00})\ket{11}.
\end{align*}
The identity element in $\cA_n^{(d)}$, is $\ket{\underbrace{0 \ldots
    0}_n}$, and the inverse, log and exponential are carried over in the
obvious way from the corresponding functions in $\cA_n^{(d)}$.

Suppose the $n$ subsystems are divided into two sets $S$ and $T$. We
write 
\begin{align}
\label{tprod}\ket{\psi}=\ket{\psi}_S \otimes \ket{\psi}_T,
\end{align}
to indicate that $\ket{\psi}$ is separable with respect to these
subsets, the order in which $S$ and $T$ appear in the tensor product
not being necessarily related to the order of their indices (e.g. we
might write $\ket{\psi}=\ket{\psi}_{13} \otimes \ket{\psi}_2$). Then
$\ket{\psi}$ can also be written in terms of the algebra product as
$\ket{\psi}=\ket{\psi}_S \ket{\psi}_T$, where $\ket{\psi}_S$ is
identified with an element of the algebra that only uses those $e_i$
with $i \in S$, and $\ket{\psi}_T$ using only those $e_i$ with $i \in
T$. From (\ref{logsum2}) we have
\begin{align}
\label{linearise} \log \ket{\psi}=\log \ket{\psi}_S+\log \ket{\psi}_T.
\end{align}
thus linearising the tensor product when the $\log$ is defined,
i.e. when the constant coefficients in the algebra do not vanish.

The coefficients $c_{i_1 \ldots i_n}$ of $\log \ket{\phi}=\sum c_{i_1
  \ldots i_n} \ket{i_1 \ldots i_n}$ are called {\em cumulants}. The
cumulant corresponding to $c_{i_1 \ldots i_n}$ is defined for
classical random variables $X_i$ as the coefficient of
$\lambda_1^{i_1}..\lambda_n^{i_n}$ in $\log\langle
e^{\sum_i\lambda_iX_i}\rangle$ \cite{Royer}. This follows by
identifying $\langle X_1^{i_1} \ldots X_n^{i_n}\rangle$ with $a_{i_1
  \ldots i_n}$ and taking $a_{0 \ldots 0}=1$. For instance, the
equivalent of
\begin{align}\label{c11}
c_{11}=(a_{11}a_{00}-a_{10}a_{01})/a_{00}^2,
\end{align}
which is the coefficient of $e_1e_2$ in (\ref{log}), is the classical
second degree cumulant $\langle X_1X_2 \rangle-\langle X_1\rangle
\langle X_2 \rangle$. The algebra $\cA_n^{(2)}$ can also be identified
with the ``moment algebra'' in \cite{AbergMitchison09} by associating
to the term $ce_{i_1} \ldots e_{i_n}$ in $\cA_n^{(2)}$ the map that
assigns to the integers $\{i_1, \ldots i_n\}$ the value $c$.

\section{Multipartite separability}

Write the state space for an $n$-party state as $(\complex^d)_1
\otimes \ldots \otimes (\complex^d)_n$. If $\pi=\{\pi_1,\pi_2,\ldots,
\pi_k\}$ is a partition of $n$, we say $\ket{\phi}$ is $\pi$-separable
if we can write
\begin{align}
\label{tensorprod}
\ket{\phi}=\bigotimes_{i=1}^k \ket{\phi}_{\pi_i},
\end{align} 
where each $\ket{\phi}_{\pi_i}$ is a state in the subspace $\bigotimes_{j
  \in \pi_i} (\complex^d)_j$. As we have seen, we can also write this
using the algebra product as
\begin{align}
\label{algebraprod} \ket{\phi}=\ket{ \phi}_{\pi_1} \ldots \ket{ \phi}_{\pi_k}.
\end{align}
From (\ref{logsum2}) we get
\begin{align}
\label{logsum} \log \ket{\phi}=\log\ket{ \phi}_{\pi_1}+ \ldots +\log\ket{ \phi}_{\pi_k}.
\end{align}
This immediately gives a characterisation of multipartite
separability. Let us say that a set of indices ${i_1 \ldots i_n}$ {\em
  splits} the partition $\pi$ if there are non-zero indices $i_j$ in
more than one subset of $\pi$.
\begin{theorem}\label{separability}
  An $n$-party pure state $\ket{\phi}$ with $a_{0 \ldots 0} \ne 0$ is
  $\pi$-separable if and only if $c_{i_1 \ldots i_n}=0$ whenever ${i_1
    \ldots i_n}$ splits $\pi$.
\end{theorem}
\begin{proof}
  Necessity follows from the fact that the cumulants with indices
  splitting $\pi$ are absent from the expansion (\ref{logsum}) of
  $\log \ket{\phi}$. Sufficiency follows by noting that, if these
  cumulants are zero, we can write $\log \ket{\phi}$ in the form
\begin{align*}
  \log \ket{\phi}=\sum_k \sum_{\{i_j=1 \ \Longrightarrow \ j \in \pi_k\}}c_{i_1i_2 \ldots i_n}\ket{i_1i_2 \ldots i_n}
\end{align*}
and exponentiating this shows $\ket{\phi}$ to be $\pi$-separable.
\end{proof}
The condition $a_{0 \dots 0} \ne 0$ reflects the special role played
by $\ket{0 \dots 0}$ as the identity in the algebra. We shall shortly
give a version of this theorem (\ref{separability-criterion})
which does not have this unpleasant restriction.

It should be emphasised that one can easily write down algebraic
conditions for a pure state to be $\pi$-separable. However, the
characterisation of Theorem \ref{separability} will turn out to
provide a useful starting point for making qubit invariants. It is
also economical, in the sense that the vanishing of the $c$'s gives
the right number of equations to define the subspace of
$\pi$-separable normalised states. Indeed, the (real parameter)
dimension of this subspace is
\[
d_\pi=\sum
\left( 2d^{|\pi_i|}-2\right),
\]
the expression in brackets counting the real and imaginary parts of
each coefficient of $\ket{\psi}_{\pi_i}$, with 2 subtracted for
normalisation and phase invariance. On the other hand, the number
$N_\pi$ of index sets that split $\pi$ is
\[
N_\pi=\left(d^n-1 \right)-\sum \left(d^{|\pi_i|}-1\right),
\]
this being the total number of index sets minus those that do not
split $\pi$, i.e. those where the 1-indices lie wholly within some
$\pi_i$. (One subtracts 1 for the all-zero sets so as not to
overcount.) But the total dimension of normalised states is
\[
d_{all}=2d^n-2
\]
and, as each equation $c_{i_1 \ldots i_n}=0$ contributes two
constraints from the vanishing of its real and imaginary parts, we
require
\[
d_\pi=d_{all} - 2N_\pi,
\]
which is readily seen to hold.

There is a result closely analogous to Theorem \ref{separability} for
maximal rank mixed states. For such a state the usual logarithm
exists, viz. $\log \rho=U \diag(\log \lambda_1, \ldots , \log \lambda_n)
U^\dagger$, where $U$ diagonalises $\rho$ with eigenvalues
$\lambda_i$.  Following the approach of Zhou \cite{Zhou09} we write
\begin{align}\label{mixedlog}
\log \rho=\sum f_{i_1 \ldots i_n}\sigma_{i_1} \ot \ldots \ot \sigma_{i_n},
\end{align}
where $\sigma_j$, $j=0,1, \ldots , d^2-1$ is a Hermitian operator
basis orthonormal with respect to the trace inner product
$\tr(AB^\dagger)/d$ (e.g. for $d=2$, the Pauli matrices). We may
choose $\sigma_0=\id$, in which case $\sigma_i$, $i \ge 1$ are trace-free.  By
the definition of the $\sigma_j$, the coefficients are given by
\begin{align}
  f_{i_1 \ldots i_n}=\frac{1}{d}\tr(\sigma_{i_1} \ot \ldots \ot
  \sigma_{i_n} \log \rho),
\end{align} 
and are real.  We now say $\rho$ is $\pi$-factorisable if
\begin{align}
\rho=\bigotimes_{i=1}^k \rho_{\pi_i},
\end{align} 
in which case
\begin{align}\label{logprod}
\log \rho=\sum_{i=1}^k \id_1 \ot \ldots \ot \log \rho_{\pi_i} \ot \ldots \ot \id_k.
\end{align} 
As above, we say that a set of indices ${i_1 \ldots i_n}$ splits $\pi$
if there are non-zero indices $i_j$ in more than one subset of
$\pi$. Then we have
\begin{theorem}\label{factorisability}
  An $n$-party mixed state $\rho$ is $\pi$-factorisable if and only if
  $f_{i_1 \ldots i_n}=0$ whenever ${i_1 \ldots i_n}$ splits $\pi$.
\end{theorem} 
\begin{proof}
  That $\tr(\sigma_{i_1} \ot \ldots \ot \sigma_{i_n} \log \rho)=0$ when
  ${i_1 \ldots i_n}$ splits $\pi$ follows from (\ref{logprod}) and
  the assumption that $\sigma_j$ is trace-free for $j \ge
  1$. Conversely, if the former condition holds, then (\ref{mixedlog})
  implies
\[
\log \rho =\sum_{i=1}^k \id_1 \ot \ldots \ot \tau_{\pi_i} \ot \ldots \ot \id_k,
\]
where $\tau_{\pi_i}$ is hermitian. Exponentiating gives
\[
\rho=\bigotimes_{i=1}^k e^{\tau_{\pi_i}},
\]
where each of the exponential factors has positive eigenvalues and can
be normalised by distributing numerical factors amongst the terms suitably.
\end{proof} 
Once again, the splitting condition gives the right number of
equations. The dimension of the subspace of $\pi$-factorisable states
is $d_\pi=\sum_i(d^{2|\pi_i|}-1)$, the number of index sets that split
$\pi$ is $N_\pi=(d^{2n}-1)-\sum_i(d^{2|\pi_i|}-1)$, and
$d_{all}=(d^{2n}-1)$ is the dimension of the space of $n$-partite
states. This time $f_{i_1 \ldots i_n}=0$ only yields one constraint,
since the $f$'s are real, so we require $d_\pi=d_{all} - N_\pi$, which
holds. 

There is a version of Theorem \ref{factorisability} that addresses
multipartite separability rather than factorisability, obtained by
borrowing the construction from classical multipartite squashed
entanglement \cite{CW03,Yang09}. We say $\rho_{.E}$ is a classical
extension of $\rho$ if $\tr_E \rho_{.E}=\rho$ and $\rho_{.E}$ has the
form $\sum_i p_i \tau^i \ot \proj{i}_E$, where $\tau^i$ are states on
$(\complex^d)_1 \otimes \ldots \otimes (\complex^d)_n$.
\begin{theorem}\label{mixed-separability}
  An $n$-party mixed state $\rho$ is $\pi$-separable if and only if
  there exists a classical extension $\rho_{.E}$ such that
  $\tr(\sigma_{i_1} \ot \ldots \ot \sigma_{i_n} \ot \proj{i}_E \log
  \rho_{.E})=0$ whenever ${i_1 \ldots i_n}$ splits $\pi$.
\end{theorem} 
\begin{proof}
  This follows the same lines as the proof of Theorem
  \ref{factorisability}, noting that $\rho$ is $\pi$-separable if and
  only if there exists a classical extension where the $\tau^i$ are
  $\pi$-factorisable.
\end{proof}

\section{Invariants for pure n-qubit states}

From now on we restrict attention to qubit states. We are interested
in polynomial invariants of the action of local unitaries, which we
take to be the group $SU(2)^n \times U(1)$. Given a state
$\ket{\psi}=\sum_{i_1 \ldots i_n} a_{i_1 \ldots i_n}\ket{i_1 \ldots
  i_n}$, we seek polynomials $I(\ket{\psi})$ in the
coefficients $a_{i_1 \ldots i_n}$ and their complex conjugates
$\bar{a}_{i_1 \ldots i_n}$ satisfying
\begin{align}
I(g\ket{\psi})=I(\ket{\psi}),
\end{align}
for any $g \in SU(2)^n \times U(1)$. Equivalently, we can express the
group action on the $i$th copy of $SU(2)$ by
\begin{align}\label{a-action}
  \rho_i(g)a_{j_1 \ldots j_i \ldots j_n}=\sum_{k_i}g_{j_i k_i} a_{j_1 \ldots k_i \ldots    j_n},
\end{align}
and extend this to any monomial $m=\prod_{q=1}^k a_{j_{q,1} \ldots
  j_{q,n}}$ by $\rho_i(g)m=\prod \left( \rho_i(g)a_{j_{q,1} \ldots
    j_{q,n}} \right)$. We then require the polynomial $I$ to be
invariant under $\rho_i(g)$ for all $i$ and $g
\in SU(2)$ as well as invariant under phase changes introduced by
$U(1)$.

With one exception, all the invariants we discuss are
real-valued. However, for systems of three or more qubits real-valued
invariants do not suffice to form a fundamental set.  Indeed, if $J$
is real, then it cannot distinguish a state from its complex
conjugate, since $J(\ket{\bar \psi})=\bar J(\ket{\psi})=
J(\ket{\psi})$. For two qubits, a state and its conjugate are
equivalent under local unitaries; this is not true for more qubits
\cite{Acinetal00}.

The cumulant $c_{i_1 \ldots i_n}$ can be made into a polynomial
$d_{i_1 \ldots i_n}$ by putting
\begin{align}\label{ddef}
  d_{i_1 \ldots i_n}=a_{0 \ldots 0}^\theta c_{i_1 \ldots i_n},
\end{align}
where $\theta$ is the the number of 1's in the set $i_1 \ldots i_n$.
This is a homogeneous polynomial of degree $\theta$ in the $a$'s. For
instance, from (\ref{c11}) we have
$d_{11}=a_{00}^2c_{11}=a_{11}a_{00}-a_{10}a_{01}$. In general, we
have the following formula for $d_{i_1 \ldots i_n}$. Let
$S=\{k|i_k=1\}$, so $|S|=\theta$ is the degree of $d_{i_1 \ldots
  i_n}$. If $A \subset S$ let $a_A$ denote $a_{j_1 \ldots j_n}$ where
$j_k=1$ if $k \in A$ and $j_k=0$ otherwise.  Then
\begin{align}\label{partitions}
d_{i_1 \ldots i_n}=\sum_\pi (-1)^{|\pi|-1}(|\pi|-1)! a_{\emptyset}^{\theta-|\pi|}\prod_{i=1}^{|\pi|}a_{\pi_i},
\end{align}
where the sum is over all partitions $\pi=\{\pi_1, \ldots ,
\pi_{|\pi|}\}$ of $S$, $|\pi|$ being the number of subsets in the
partition.

The action of local unitaries on $d_{11}$ is given by
\begin{align}
\rho_i(g)d_{11}=\Delta d_{11}, \label{nice}
\end{align}
for $i=1,2$, where $\Delta=\det g$. To remove the dependence on the
phase $\Delta$, i.e. to get invariance under the action of $U(1)$, we
take
\begin{align}
I_{11}=|d_{11}|^2.
\end{align}
This gives us a local unitary invariant, and
\begin{align}\label{2inv}
I_{11}=|d_{11}|^2.
\end{align}
This gives us a local unitary invariant, and there is just one such
invariant for normalised 2-qubit states. The transformation of
cumulants for three or more qubits is more complicated (Theorem
\ref{action}). However, we can always get an invariant by integrating
the squared modulus of $d_{i_1 \ldots i_n}$ over the group $SU(2)^n$,
and this prompts the following:
\begin{definition}
\begin{align}\label{Idef}
  I_{i_1 \ldots i_n}=\gamma_{n,\theta}\int_{SU(2)^n} |\rho_1(g_1) \ldots \rho_n(g_n)d_{i_1 \ldots
    i_n}|^2dg_1 \ldots dg_n.
\end{align} 
Here the integral is over the Haar measure, and the constant factor
$\gamma_{n,\theta}=(\theta+1)^{n-\theta}(\theta-1)^\theta$ is
introduced for later convenience.
\end{definition}

We now explore the properties of these invariants. First, we note that
we can give a more satisfying, basis independent, version of Theorem
\ref{separability}:
\begin{theorem}[Separability criterion]\label{separability-criterion}
  An $n$-qubit state $\ket{\psi}$ is $\pi$-separable if and only if
  $I_{i_1 \ldots i_n}=0$ whenever ${i_1 \ldots i_n}$ splits $\pi$.
\end{theorem} 
\begin{proof}
  Suppose $\ket{\psi}$ is $\pi$-separable, so $\ket{\psi}=\bigotimes
  \ket{\psi}_{\pi_i}$. If $a_{0 \ldots 0} \ne 0$, Theorem
  \ref{separability} says that $c_{i_1 \ldots i_n}=0$ and hence
  $d_{i_1 \ldots i_n}=0$. If $a_{0 \ldots 0} = 0$, define
  $\ket{\psi}_x=\bigotimes (\ket{\psi}_{\pi_i} + x\ket{0 \ldots
    0}_{\pi_i})$.  If $x \ne 0$, $a_{0 \ldots 0} \ne 0$, so $d_{i_1
    \ldots i_n}=0$ for this state. By continuity, $d_{i_1 \ldots
    i_n}=0$ for $x=0$, i.e. for the original $\ket{\psi}$. Thus
  $I_{i_1 \ldots i_n}=0$. Conversely, if $I_{i_1 \ldots i_n}=0$, then
  $d_{i_1 \ldots i_n}(g\ket{\psi})=0$ for all $g$, and for some $g$ we
  must have $a_{0 \ldots 0} \ne 0$, so Theorem \ref{separability} can
  be applied.
\end{proof}

Next we define the action of elements $g \in SU(2)$ on our cumulant
polynomials. We will need some notation.
\begin{definition}
  Define $\cS_i$ by 
\begin{align}\label{Adef}
  \cS_i a_{l_1 \ldots 0_i \ldots l_n}&=a_{l_1 \ldots 1_i \ldots l_n},\\
  \cS_i a_{l_1 \ldots 1_i \ldots l_n}&=0.
\end{align}
Now, given a monomial $m=\prod_{q=1}^\theta a_{i_{1,q} \ldots
  i_{n,q}}$ of degree $\theta$, define $\cR_{i,k}$ to be the
coefficient of $x^k$ in $\prod (\id+xS_i)a_{i_{1,q} \ldots
  i_{n,q}}$. Extend this definition by linearity to any homogeneous
polynomial of degree $\theta$.
\end{definition}
For instance, for $\theta=3$ we have
\[
\cR_{i,1}m =\left( \cS_i \id \id+\id \cS_i \id + \id \id \cS_i \right)m.
\]
Note that, by virtue of the symmetry, this definition does not depend
on the order of the $a$'s in $m$. We can think of $\cR_{i,k}$ as being
like a raising operator (hence the `R'). For example
\begin{align}
\label{R0} \cR_{3,0}d_{110}&=a_{110}a_{000}-a_{100}a_{010},\\
\label{R1} \cR_{3,1}d_{110}&=a_{111}a_{000}+a_{110}a_{001}-a_{101}a_{010}-a_{100}a_{011},\\
\label{R2} \cR_{3,2}d_{110}&=a_{111}a_{001}-a_{101}a_{011},
\end{align}

\begin{lemma}\label{xi-sum}
  If $l_i=1$ and at least one other index in $d_{l_1 \ldots l_n}$ is
  1, then $\cR_{i,\theta}d_{l_1 \ldots l_n}=\cR_{i,\theta-1}d_{l_1
    \ldots l_n}=0$.
\end{lemma}
\begin{proof}
  Writing $d$ for $d_{l_1 \ldots l_n}$, $\cR_{i,\theta}d=0$ because we
  cannot add $\theta$ 1's to the $\theta-1$ $a$'s that originally had
  zeros in position $i$.  $\cR_{i,\theta-1}d$ has 1's at position $i$
  in every a, so we can write $\cR_{i,\theta-1}d=d(\ket{\psi})$, where
  $\ket{\psi}$ is a state with $a_{l_1 \ldots 0_i \ldots l_n}=a_{l_1
    \ldots 1_i \ldots l_n}$. This means that $\ket{\psi}$ factorises
  as $\left(\frac{\ket{0}+\ket{1}}{\sqrt{2}}\right)_i \otimes
  \ket{\psi}_{\cN-i}$.  We now invoke Theorem
  \ref{separability-criterion}, since $\{l_1, \ldots, l_n\}$ splits
  the partition $\{i\},\{\cN-i\}$.
\end{proof}
\begin{lemma}\label{Ld}
  If we write $L_id$ for the result of replacing all the 1's in
  position $i$ in $d$ by 0, and if $l_i=1$, then $L_id=0$.
\end{lemma}
\begin{proof}
  The same argument as in the proof above shows that
  $L_id=d(\ket{\psi})$ where $\ket{\psi}$ factorises.
\end{proof}
\begin{theorem}[Action of local unitaries.]\label{action}
  Let $l_i$ be one of the indices of $d_{l_1, \ldots, l_n}$, and let
  $g=\left(\begin{array}{ccc}u & v\\ w &
      z \end{array}\right)$. Assume that the index set $l_1, \ldots,
  l_n$ contains at least two 1's. Then if $l_i=0$,
  \begin{align}\label{i=0}
    \rho_i(g)d=\sum_{k=0}^\theta u^{\theta-k}v^k \cR_{i,k}d,
  \end{align}
    and if $l_i=1$,
  \begin{align}\label{i=1}
    \rho_i(g)d=\Delta \sum_{k=0}^{\theta-2} u^{\theta-k-2}v^k \cR_{i,k}d,
  \end{align}
where $d$ stands for $d_{l_1, \ldots, l_n}$. 
\end{theorem}
\begin{proof}
  When $l_i=0$ and $m=\prod_{q=1}^\theta a_{j_{q_1} \ldots
    j_{q_n}}$,
\begin{align}
  \rho_i(g) m=\prod_{q=1}^\theta \left( ua_{j_{q,1} \ldots 0_i
 \ldots j_{q,n}}+va_{j_{q,1} \ldots 1_i \ldots j_{q,n}}\right),
\end{align}
and it is clear that the coefficient of $u^\theta$ is the original
monomial $m$, and the coefficient of $u^{\theta-k}v^k$ is
$\cR_{i,k}d$.

When $l_i=1$, 
\begin{align}
  \rho_i(g)m=\left( wa_{j_{p,1} \ldots 0_i \ldots j_{p,n}}+za_{j_{p,1} \ldots
      1_i \ldots j_{p,n}}\right)\ \prod_{q=1,q \ne p}^\theta \left( u a_{j_{q,1} \ldots 0_i \ldots
      j_{q,n}}+v a_{j_{q,1} \ldots 1_i \ldots j_{q,n}}\right),
\end{align}
where $a_{j_{p,1} \ldots \ldots j_{p,n}}$ is the unique $a$ in $m$
that has a 1 at position $i$. The coefficient of $zv^{\theta-1}$
is $\cR_{i,\theta-1}d$, which is zero by Lemma \ref{xi-sum}, and for $k
\le \theta-2$ the coefficient of $u^{\theta-k-2}v^k(uz)$ is
$cP_{i,k}d$. The coefficient of $u^{\theta-1}w$ is $Ld$, which is
zero by Lemma \ref{Ld}, and for $k \ge 0$, the coefficient of
$u^{\theta-k-2}v^k(vw)$ is
$\cR_{i,k+1}(Ld)-\cR_{i,k}d=-\cR_{i,k}d$. Since
$\Delta=uz-vw$, the theorem follows.
\end{proof}
The polynomials $\cR_{i,k}d$ appearing in equations (\ref{i=0}) and
(\ref{i=1}) form an irreducible representation for $SU(2)$. To see
this, let $h=\left(\begin{array}{ccc}a & b\\ c &
    d \end{array}\right)$. Then, for $l_i=0$, applying $\rho_i(h)$ to
(\ref{i=0}) gives
\begin{align}
  \label{rep} \sum_{k=0}^\theta u^{\theta-k}v^k
  \left(\rho_i(h)\cR_{i,k}d\right)=\rho_i(h)\rho_i(g)d=\rho_i(gh)=\sum_{k=0}^\theta
  (ua+vc)^{\theta-k}(ub+vd)^k \cR_{i,k}d,
\end{align}
and equating coefficients of $u^{\theta-k}v^k$ in the left- and
right-hand sides of (\ref{rep}) gives the action of $h$ at position
$i$ on all the polynomials $\cR_{i,k}d$ and thus defines the
representation matrix, which can easily be recognised as the symmetric
representation with Young diagram $(\theta)$. When $l_i=1$, the same
argument applied to (\ref{i=1}) shows that we obtain the
representation with Young diagram $(\theta-1,1)$. Classically,
cumulants were called ``half-invariants'' \cite{Thiele1,Dressel40}
because, if $c$ is the classical joint cumulant in the random
variables $X_1 \ldots X_n$, then $c$ is mapped to $z^nc$ when $X_i$ is
transformed by the affine map $X_i \to zX_i+w$. The equivalent of an
affine map in our setting is the group of matrices
$\left(\begin{array}{ccc}1 & 0\\ w & z \end{array}\right)$ with $z \ne
0$. From Theorem \ref{action} we see that the representation becomes
1-dimensional, sending $d \to z^\theta d$.

\begin{theorem}[Formula for invariants]\label{explicit}
If the index set $i_1, \ldots , i_n$ contains at least two 1's, then
\begin{align}\label{formula}
  I_{i_1, \ldots, i_n}=\sum_{k_1, \ldots , k_n} \left|\prod_{p=1}^n
    \left( \alpha^{i_p}_{k_p} \cR_{i_p,k_p} \right) d_{i_1, \ldots, i_n} \right|^2.
  \end{align}
  Here, if $i_p=0$, $k_p$ ranges from 0 to $\theta$, and
  $\alpha^0_{k_p}=\binom{\theta}{k_p}^{-1}$. If $i_p=1$, $k_p$
  ranges from 0 to $\theta-2$, and
  $\alpha^1_{k_p}=\binom{\theta-2}{k_p}^{-1}$
\end{theorem}
\begin{proof}
  Combining (\ref{Idef}) and Theorem \ref{action} we get terms from
  $\int |\rho_i(g)d|^2dg$ such as $\int |u|^{2(\theta-k)}|v|^{2k}dg
  |\cR_{i,k}d|^2$, and the integral can be calculated by using
  Schur's lemma:
\begin{align}\label{schur1}
\binom{p+q}{p} \int |u|^{2p} |v|^{2q} dg&=\bra{\psi_{p,q}}\int g^{\otimes (p+q)} \proj{0}^{\otimes (p+q)} (g^\dagger)^{\otimes (p+q)} dg \ket{\psi_{p,q}}\\
&=\dim \Sym^{p+q}(\complex^2)^{-1}\bra{\psi_{p,q}}P_{Sym} \ket{\psi_{p,q}}\\
&=\dim \Sym^{p+q}(\complex^2)^{-1}\\
\label{schur2}&=(p+q+1)^{-1}.
\end{align}
Here $P_{Sym}$ is the projector onto the symmetric representation, and
$\ket{\psi_{p,q}}$ is the normalised weight vector
$\binom{p+q}{p}^{-1}\sum_\sigma \ket{i_{\sigma(1)} \ldots
  i_{\sigma(p+q)}}$, with $i_1=\ldots =i_p=0$, $i_{p+1}= \ldots
=i_{p+q}=1$, and with the sum taken over all permutations in
$S_{p+q}$. There are $n-\theta$ indices $i_p$ that are zero, where
$(p+q+1)=\theta+1$, and $\theta$ indices $i_p$ that are 1, where
$(p+q+1)=\theta-1$. These terms therefore cancel the constant
$\gamma_{n,\theta}=(\theta+1)^{n-\theta}(\theta-1)^\theta$ in
(\ref{Idef}), and the $\alpha^{i_p}_{k_p}$'s come from the factor
$\binom{p+q}{p}$ in (\ref{schur1}).

There are also cross-terms in $\int |\rho_i(g)d|^2dg$, but these have
the form $\int u^{(\theta-k)}v^k \overline{u}^{(\theta-j)}\overline{v}^j dg
  \left( \cR_{i,k}d \overline{\cR_{i,j}d}\right)$, and
\begin{align}\label{schur3}
\int u^{(\theta-k)}v^k \overline{u}^{(\theta-j)}\overline{v}^j dg
=\langle \psi_{\theta-k,k} \vert \psi_{\theta-j,j} \rangle,
\end{align}
which vanishes for $j \ne k$ since distinct weight vectors are orthogonal.
\end{proof}

Next we show that the invariants are algebraically independent. This
means there can be no non-trivial polynomial relations between
them. In the case of our cumulant-derived invariants, we can quickly
eliminate certain types of functional dependence by using the
separability criterion, Theorem \ref{separability-criterion}. For
instance, we cannot have $I_{110}=f(I_{111},I_{101},I_{011})$ for any
function $f$ because all the invariants on the right-hand side are
zero for any $\{12\}\{3\}$-separable state, whereas $I_{110}$ can take
a range of values according to the choice of the state. But we cannot
use the separability criterion to rule out a relation of the form
$I_{111}=f(I_{011},I_{101},I_{110})$, say. It turns out, however, that
we can rule out such relations. Indeed, we have
\begin{theorem}[Algebraic independence]
A polynomial relationship between the invariants  $I_{i_1 \ldots i_n}$
\begin{align}\label{dependence}
\sum \alpha_{i_{1,1} \ldots
    i_{k,n}}I^{j_1}_{i_{1,1} \ldots i_{1,n}} \ldots I^{j_k}_{i_{k,1} \ldots i_{k,n}} =0,
\end{align}
can only hold if each $\alpha_{i_{1,1} \ldots i_{k,n}}=0$.
\end{theorem}
\begin{proof}
  Consider the neighbourhood of the fully separable state $\ket{0
    \ldots 0}$, where the $a_{i_1 \ldots i_n}$'s with not all $i_j=0$
  are small. From (\ref{formula}) and (\ref{partitions}) we see that
  the expansion of the invariants in lowest degree terms in these
  small variables only has contributions from the term $a_{j_1 \ldots
    j_n}a^{\theta-1}_{0 \ldots 0}$ in $d_{j_1 \ldots j_n}$, and from
  the corresponding term where operators $\cR_{i,1}$ applied to
  positions where the index $j_i=0$. Thus, assuming that the indices
  $i_1, \ldots , i_n$ include at least two 1's, we have
  \begin{align}\label{map}
  I_{i_1 \ldots i_n} \approx x^{\theta-1}_{0 \ldots 0}\sum^1_{i_{t_1}=0} \ldots \sum^1_{i_{t_{n-\theta}}=0} \left( \theta^{-\sum_j i_{t_j}} \right)x_{i_1 \ldots i_{t_1} \ldots i_{t_2} \ldots i_n},
  \end{align}
  where $x_{j_1 \ldots j_n}=|a_{j_1 \ldots j_n}|^2$, and ${t_1} \ldots
  {t_{n-\theta}}$ denote the positions of indices in $i_1, \ldots ,
  i_n$ that are zero. We also have the special case
\[
I_{10 \ldots 0}=\sum^1_{i_1=0} \ldots \sum^1_{i_n=0} x_{i_1 \ldots i_n}.
\]
Let us write $I_1=I_{10 \ldots 0}-x_{0 \ldots 0}$, and
$x_1=x_{10\ldots 0}+x_{010 \ldots 0} + \ldots x_{0\ldots 01}$. Then
the lowest degree expansion defines an invertible map $A$. For
instance, for $n=3$ this map is given by
\[
A=\bordermatrix{
       & I_1 &I_{110} & I_{101} & I_{011} & I_{111}\cr
x_1    & 1   & 0     & 0   & 0& 0  \cr
x_{110} & 1   & x_{000}& 0 & 0& 0  \cr 
x_{101} & 1   & 0     & x_{000}& 0& 0  \cr
x_{011} & 1   & 0     & 0& x_{000}& 0  \cr
x_{111} & 1   & \frac{1}{2}x_{000} & \frac{1}{2}x_{000}& \frac{1}{2}x_{000}& x^2_{000} \cr
}
\]
This extends to a map $A^{poly}$ from polynomials in the $I$'s to
polynomials in the $x$'s. Since $A$ is invertible, the same is true of
$A^{poly}$. If we now take the lowest degree terms in
(\ref{dependence}), they map to a polynomial in the $x$'s whose
coefficients must be zero, since the $x$'s are independent
variables. Applying the inverse of $A^{poly}$, we deduce that the
corresponding coefficients in (\ref{dependence}) are zero. Looking at
the next highest degree terms in (\ref{dependence}), we again deduce
that their coefficients are zero. And so on for the whole
$I$-polynomial.
\end{proof}

We can therefore conclude that we get $2^n-n$ algebraically
independent invariants from cumulants, namely one invariant, the
squared-amplitude, from the first-degree cumulants, and $2^n-n-1$ from
the second and higher degree cumulants. On the other hand, $D(n)$, the
dimension of the space of orbits, is $2^{n+1}-(3n+1)$ for $n \ge 3$,
\cite{Carteret00}. Thus asymptotically the number of cumulant-based
invariants is $\frac{1}{2}D(n)$.

Geometrically speaking, algebraically independent invariants give
information about independent directions in orbit space, and $D(n)$ is
the number of invariants in a maximal algebraically independent
set. However, a maximal algebraically independent set will generally
not suffice to characterize orbits fully: there will be a finite set
of states that take the same values on all of the invariants in such a
set, and additional invariants are needed to distinguish amongst these
states. Thus a fundamental set, while it certainly contains a maximal
algebraically independent set, will contain additional invariants, and
generally contains a large number of such invariants.

\begin{example}\label{example3}
Consider the case of three qubits. From
(\ref{formula}) we find
\begin{align}
  I_{110}=|\cR_{3,0}d_{110}|^2+\frac{1}{2}|\cR_{3,1}d_{110}|^2+|\cR_{3,2}d_{110}|^2,
\end{align}
which is a polynomial invariant of degree 4. The terms are given
explicitly by (\ref{R0}), (\ref{R1}) and (\ref{R2}). There are two
other 4th degree invariants, $I_{101}$ and $I_{011}$ obtained by
permuting the indices. We also have
\begin{align}
  I_{111}=\sum_{i,j,k=0}^1 |\cR_{1,i}\cR_{2,j}\cR_{3,k}d_{111}|^2,
\end{align}
which is of degree 6. The theorem excludes the case where there is a
single 1 in the index set, but in that case Definition \ref{Idef} just
gives the squared amplitude $\braket{\psi}{\psi}$, irrespective of the
position of the 1.

So we get altogether five cumulant-based invariants. We can relate
these to well-known sets of three-qubit invariants
\cite{LindenPopescu98,Sudbery01,Luque07}. For instance, in the list
given in Luque et. al.  \cite{Luque07} we can make the following
identifications:
\begin{align}
\label{100} I_{100} &\leftrightarrow A_{111},\\
\label{110} I_{110} &\leftrightarrow B_{002},\\
\label{111} I_{111} &\leftrightarrow C_{111},
\end{align}
including permutations of (\ref{110}). Sudbery's list \cite{Sudbery01}
includes $J_1=\braket{\psi}{\psi}$, the three fourth-degree
polynomials $J_2=\tr \rho_3^2$, $J_3=\tr \rho_2^2$, $J_4=\tr
\rho_1^2$, and the sixth-degree polynomial $J_5=3\tr\left[(\rho_1
  \otimes \rho_2)\rho_{12}\right]-\tr(\rho^3_1)-\tr(\rho_2^3)$, where
$\rho_1=\tr_{23}\proj{\psi}$, $\rho_2=\tr_{13}\proj{\psi}$,
$\rho_3=\tr_{12}\proj{\psi}$, and $\rho_{12}=\tr_3\proj{\psi}$. These
are related to the cumulant-based invariants by
\begin{align*}
I_{100}&=I_{010}=I_{001}=J_1,\\
4I_{110}&=J_1^2+J_2-J_3-J_4,\\
4I_{101}&=J_1^2+J_3-J_2-J_4,\\
4I_{011}&=J_1^2+J_4-J_2-J_3,\\
6I_{111}&=5J_1^3-3J_1\left(J_2+J_3+J_4\right)+4J_5.
\end{align*} 

Since $D(3)=6$, one extra invariant is needed to make a maximal
  algebraically independent set. In Sudbery's list \cite{Sudbery01}
  this is $J_6=|$Det$(\ket{\psi})|^2$, where 
\begin{align}\label{Det}
\mbox{Det}(\ket{\psi})=
  a_{ijk}a_{i^\prime j^\prime m}a_{npk^\prime}a_{n^\prime p^\prime
    m^\prime}\epsilon_{i i^\prime}\epsilon_{j j^\prime}\epsilon_{k
    k^\prime}\epsilon_{m m^\prime}\epsilon_{n n^\prime}\epsilon_{p
    p^\prime}
\end{align}
is the hyperdeterminant \cite{Miyake}. $J_6$ is closely related to the
3-tangle \cite{Coffman00}, $\tau=2|$Det$(\ket{\psi})|$.

To make a fundamental set, we need to add one more invariant to this
maximal algebraically independent set, and Grassl \cite{Grassl02} has
identified a suitable 12th degree polynomial. This is a complex
invariant, and is necessary for distinguishing a state from its
complex conjugate \cite{Acinetal00}.

\end{example}

\section{The hierarchical structure of invariants}

The 3-qubit invariants $I_{011}, I_{101}$ and $I_{110}$ are closely
related to the 2-qubit invariant $I_{11}$. We obtain $I_{110}$ from
$I_{11}$ by adding a `0' index in the third position to $I_{11}$ and
twirling. More generally, we wish to make an $(n+1)$-qubit invariant
from an $n$-qubit invariant $J$. Phase invariance requires every term
in $J$ to be a product of an equal number of $a$'s and $\bar a$'s.
Given a monomial $m=\prod a_{j_1 \ldots j_n}\bar{a}_{k_1 \ldots k_n}$
in $J$, of degree $\theta$ in the $a$'s and $\bar a$'s, we define
\begin{align}\label{twirl-lift}
  m_0=(\theta+1)\int_{SU(2)} \rho_i(g)\prod a_{j_1 \ldots 0_i \ldots
    j_n}\bar{a}_{k_1 \ldots 0_i \ldots k_n} dg.
\end{align}
We refer to this process as {\em lifting} and use a subscript `0' to
denote a lifted index. Table \ref{lifts} shows lifts of $I_{11}$ and
$I_{111}$ up to five qubits. We say that an invariant together with
its lifts constitute a {\em family}.

Lifting can also be understood in terms of tracing-out operations. Let
us write a mixed state of $n$ qubits as
\begin{align}\label{mixed}
\rho=\sum_{i_1, \ldots i_n; j_1, \ldots j_n}a_{i_1, \ldots i_n}^{j_1, \ldots j_n}\bigotimes_{r=1}^n \ket{i_r}\bra{j_r}
\end{align}
The following rules, applied to each monomial, enable one to
interconvert between any pure state invariant $J$ of degree $\theta$
in the $a$'s and $\bar a$'s, and a mixed state invariant $\hat J$ of
degree $\theta$ in the coefficients of (\ref{mixed}):
\begin{align}
\label{JtohatJ}  &J \to {\hat J}: \ \ \ \prod_{r=1}^\theta a_{i^r_1, \ldots i^r_n}\prod_{s=1}^\theta \bar a_{j^s_1, \ldots j^s_n} \longrightarrow \frac{1}{\theta!}\sum_{\sigma \in S_\theta} \prod_{r=1}^\theta a_{i^r_1, \ldots i^r_n}^{j^{\sigma(r)}_1, \ldots j^{\sigma(r)}_n}\\
  &{\hat J} \to J:\ \ \ \prod_{r=1}^\theta a_{i^r_1, \ldots , i^r_n}^{j^r_1,
    \ldots , j^r_n} \longrightarrow \prod_{r=1}^\theta a_{i^r_1,
    \ldots , i^r_n} \prod_{r=1}^\theta {\bar a}_{j^r_1, \ldots ,
    j^r_n}
\end{align}
The map (\ref{JtohatJ}) is given in \cite{LLW04}, but without the
  averaging over all permutations of pairings of upper and lower
  indices, which is essential for the following: 
\begin{proposition}
  If $\rho_i(g)J=J$ then $\rho_i(g) {\hat J}\rho_i(g)^\dagger=\hat
  J$. Thus given any pure state invariant $J$, ${\hat J}$ is a mixed
  state invariant.
\end{proposition}
\begin{proof}
The statement is equivalent to the commuting of the following diagram:
\[\begin{array}{rcc}J&\xrightarrow{\rho_i(g)} & J\\
\downarrow && \downarrow\\
{\hat J}&\xrightarrow{\Psi(\rho_i(g))}& {\hat J}
\end{array} \]
where $\Psi(x){\hat J}=x{\hat J}x^\dagger$.
\end{proof}
The following is immediate:
\begin{proposition}\label{purestate}
\[
{\hat J}(\proj{\psi})=J(\ket{\psi}).
\]
\end{proposition}
\begin{proposition}\label{tracedef}
  If $\ket{\psi}$ is an $n$-qubit state and $J_{p_1 \ldots p_k}$ a polynomial
  invariant, then
\[
J_{ p_1 \ldots p_k 0^{n-k}}(\ket{\psi})={\hat J}_{p_1 \ldots p_k}(\tr_{n-k} \proj{\psi}),
\]
where $0^{n-k}$ means a string of $(n-k)$ `0' indices, and $\tr_{n-k}$
means tracing out the last $(n-k)$ systems.
\end{proposition}
\begin{proof}
  Let us see how this works in a simple case. The generalisation is
  then straightforward. So consider the monomial
  $m=a_{i_1i_2}a_{j_1j_2}\bar a_{k_1k_2}\bar a_{l_1l_2}$ of a 4th
  degree 2-qubit invariant $J_{p_1p_2}$.  Under the map
  (\ref{JtohatJ}) $m$ becomes
  $\widehat{m}=\frac{1}{2}\left(a_{i_1i_2}^{k_1k_2}a_{j_1j_2}^{l_1l_2}+a_{i_1i_2}^{l_1l_2}a_{j_1j_2}^{k_1k_2}
  \right)$.  The monomial $m$ lifts, by (\ref{twirl-lift}), to
\begin{align}\label{m0}
  m_0=(\theta+1)\int \rho_3(g)a_{i_1i_20}a_{j_1j_20}\bar a_{k_1k_20}\bar
  a_{l_1l_20}dg
\end{align}
in $J_{p_1p_20}(\ket{\psi})$ for a 3-qubit state $\ket{\psi}$, and
$(\theta+1)=3$. We therefore wish to show that $m_0=\widehat m$, where
the coefficients $a_{i_1i_2}^{k_1k_2}$, etc., in $\widehat m$ come
from $\tr_3 \proj{\psi}$. This gives
 \begin{align}\label{mhat}
\widehat{m}&=\frac{1}{2}\left( a_{i_1i_2 0}\bar
  a_{k_1k_2 0}+a_{i_1i_2 1}\bar a_{k_1k_2 1} \right)\left(
  a_{j_1j_2 0}\bar a_{l_1l_2 0}+a_{j_1j_2 1}\bar a_{l_1 l_2 1}
\right)\\
\nonumber &+\frac{1}{2}\left( a_{i_1i_2 0}\bar a_{l_1l_2 0}+a_{i_1i_2 1}\bar
  a_{l_1l_2 1} \right)\left( a_{j_1j_2 0}\bar
  a_{k_1k_2 0}+a_{j_1j_2 1}\bar a_{k_1k_2 1} \right),
\end{align}
and from (\ref{m0}) we have
\begin{align}\label{m0expand}
m_0=\int \left(ua_{i_1i_2 0}+va_{i_1i_2 1}\right)\left(ua_{j_1j_2 0}+va_{j_1j_2 1}\right)\overline{\left(ua_{k_1k_2 0}+va_{k_1k_2 1}\right)}\ \ \overline{\left(ua_{l_1l_2 0}+va_{l_1l_2 1}\right)}.
dg
\end{align}
Comparing (\ref{mhat}) and (\ref{m0expand}), and using (\ref{schur1}),
(\ref{schur2}) and (\ref{schur3}), we see that $m_0=\widehat m$. This
argument is unchanged when $J$ is an $n$-qubit invariants (we just
permute notationally more cumbersome blocks of indices). When the
degree, $2\theta$, is arbitrary, a term in $\widehat{m}$ that has $p$
0's in the $a$'s in the lifted index position occurs $p!(\theta-p)!$
times in the generalisation of (\ref{mhat}). With the normlising
factor $1/\theta!$ from (\ref{JtohatJ}), we obtain the coefficient
$\int|u|^{2p}|v|^{2\theta-2p}$ of the corresponding term in $m_0$, as
given by (\ref{m0expand}).
\end{proof}
Proposition \ref{tracedef} provides an alternative definition of
lifting via tracing-out of subsystems. Combining this Proposition with
Proposition \ref{purestate} we get

\begin{corollary}
If $\ket{\psi}=\ket{\mu}_k\otimes \ket{\nu}_{n-k}$ then
\[
J_{i_1 \ldots i_k0^{n-k}}(\ket{\psi})=J_{i_1 \ldots i_k}(\ket{\mu}).
\]
\end{corollary}

The third, equivalent, definition of lifting comes from a technique
called transvection, invented by Cayley in the 19th century heyday of
invariant theory. Transvection is a useful device for generating
invariants; understanding it will enable us to interpret the 4-qubit
invariants given in \cite{Luque07} in the language used here.

The {\em fundamental form} for a $n$-qubit state is the polynomial $f$
in the a's and the variables $x^{(j)}_0$, $x^{(j)}_1$, for $1 \le j
\le n$ given by
\begin{align}
f=\sum_{i_1, \ldots , i_n} a_{i_1 \ldots i_n}x^{(1)}_{i_1} \ldots x^{(n)}_{i_n}
\end{align}
If we let $g \in SU(n)$ act on the $i$th index of $a$'s by the usual
transpose action (\ref{a-action}) and upon the $x^{(i)}$'s via the
inverse, $g^\dagger$, then one easily checks that $\rho(g)f=f$. More
generally, a {\em covariant} of weight $q$ is a polynomial $p$ in the
$a$'s and $x$'s satisfying
\begin{align}
\rho_i(g)p=\Delta_i^qp.
\end{align}
Given two covariants, $p$ and $q$, we define $\langle p,q \rangle$, by
$\langle \mu(a)|\nu(a) \rangle =\mu(a)
\overline{\nu(a)}$ for expressions in the $a$'s, and, for each $i$,
\begin{align}\label{inner-product}
\langle (x^{(i)}_0)^{p_i}(x^{(i)}_1)^{q_i}| (x^{(i)}_0)^{p^\prime_i}(x^{(i)}_1)^{q^\prime_i}\rangle=\delta_{p_i,p^\prime_i}\delta_{q_i,q^\prime_i}p_i!q_i!.
\end{align}
This is sometimes called the {\em derivative inner product} because we
obtain it by setting $x^{(i)}_j$ on the lefthand side to
$\partial/\partial x^{(i)}_j$ and applying these derivatives to
(unchanged) $x$'s on the righthand side.

\begin{table}
  \caption{Some 2, 3, 4, and 5-qubit invariants related by lifting.
\label{lifts}}
\begin{center}
\begin{tabular}{ |c|c|c|c|}
\hline
{\bf 2 qubits} & {\bf 3 qubits} & {\bf 4 qubits} & {\bf 5 qubits}\\
\hline
\hline
$I_{10}$ & $I_{100}$ & $I_{1000}$ & $I_{10000}$\\ 
\hline
$I_{11}$ & $I_{110}, I_{101}, I_{011}$ & $I_{1100}, I_{1010},
I_{0110}, I_{1001}, I_{0101}, I_{0011}$ & $I_{11000}, I_{10100}$, etc. (10) \\
\hline
- & - & $G_{1111}$ & $G_{11110}, G_{11101}, G_{11011}, G_{10111}, G_{01111}$\\
\hline
- &$I_{111}$ & $I_{1110}, I_{1101}, I_{1011}, I_{0111}$ & $I_{11100}, I_{11010}$, etc. (10) \\
\hline
- & $H_{222}$ & $H_{2220}, H_{2202}, H_{2022}, H_{0222}$ &
$H_{22200}, H_{22020}$, etc.  (10)\\
\hline
\end{tabular}
\end{center}
\end{table}

\begin{table}
\caption{Sixteen of the nineteen four-qubit invariants in \cite{Luque07} written in transvectant notation.
\label{transvectants}}
\begin{center}
\begin{tabular}{ |c|c|c|c|}
\hline
Invariants & Corresponding covariants & number & degree\\
\hline
\hline
$I_{1000}$ & $f$&1 &2\\
\hline
$G_{1111}$ &$(f,f)^{1111}$&1&4\\ 
\hline
$I_{1100}$ & $(f,f)^{1100}$&6&4\\
\hline
$I_{1110}$ & $(f,(f,f)^{1100})^{0010}$&4&6\\
\hline
$H_{2220}$ & $(f,(f,(f,f)^{1100})^{0010})^{1110}$ &4&8\\
\hline
\end{tabular}
\end{center}
\end{table}

If $p$ is a covariant, then $\langle p,p \rangle$ is an invariant, so
any means of generating covariants also supplies us with invariants.
Transvection is just such a means. Given two covariants,
$p(x^{(i)}_j)$, $q(y^{(i)}_j)$, define the {\em transvectant} by
\begin{align}\label{transvectant-def}
  (p,q)^{i_1 \ldots i_k}= \Omega_{i_1} \ldots
  \Omega_{i_k}(pq)\Bigm\vert_{y \to x} \mbox{ where }
  \Omega_iX=\frac{\partial}{\partial
    x^{(i)}_0}\frac{\partial}{\partial
    y^{(i)}_1}-\frac{\partial}{\partial
    x^{(i)}_1}\frac{\partial}{\partial y^{(i)}_0}.
\end{align}
The vertical bar indicates that, after applying the differential
operators $\Omega_i$, we change the $y$'s to $x$'s, so $(p,q)^{i_1
  \ldots i_k}$ is a polynomial in $a$'s and $x$'s. A classical theorem
\cite{Olver} asserts that, for any binary indices $i_1 \ldots i_k$,
$(p,q)^{i_1 \ldots i_k}$ is a covariant if $p$ and $q$ are. Starting
with the fundamental form, we can build up a wealth of covariants $p$
and derive invariants $\langle p|p \rangle$ from them (see Table
\ref{transvectants}).
\begin{example}
  For two-qubit states, $f=\sum a_{ij}x_i^{(1)}x_j^{(2)}$. Take
  $p=(f,f)^{11}$. Then $p=d_{11}$ and $\langle p|p
  \rangle=|d_{11}|^2=I_{11}$. For three-qubit states, $f=\sum
  a_{ijk}x_i^{(1)}x_j^{(2)}x_k^{(3)}$.  Take
  $\iota_{110}=(f,f)^{110}$. Applying (\ref{transvectant-def}) we get
\[
\iota_{110}=(f,f)^{110}=d_{110}(x^{(3)}_0)^2+\cR_{3,1}d_{110}(x^{(3)}_0x^{(3)}_1)+\cR_{3,2}d_{110}(x^{(3)}_1)^2.
\]
Using the derivative formula for the inner product
(\ref{inner-product}) we find that $\langle \iota_{110}|\iota_{110}
\rangle=4I_{110}$.
\end{example}
More generally, we have the following result:
\begin{theorem}\label{translate}
Let
\[
\iota_{1^k0^{n-k}}=(f,(f, \ldots (f,f)^{110 \ldots 0})^{001 \ldots 0} \ldots )^{0 \ldots 01 0^{n-k}}.
\]
Then 
\[
\langle \iota_{1^k 0^{n-k}}|\iota_{1^k 0^{n-k}} \rangle=\xi I_{1^k 0^{n-k}},
\]
where $\xi=4((k-2)!)^k(k!)^{n-k}$.
\end{theorem}
\begin{proof}
  Consider first the terms in $\iota_{1^k 0^{n-k}}$ where the subscript
  in every $x$ is 0.The first transvectant step, $(f,f)^{110 \ldots
    0}$, yields terms 
\begin{align}\label{step1}
\sum_{i_3 \ldots i_n; j_3 \ldots j_n} \left(a_{11i_3 \ldots i_n}a_{00j_3 \ldots j_n}-a_{10i_3 \ldots i_n}a_{01j_3 \ldots j_n} \right) x^{(3)}_{i_3}x^{(3)}_{j_3} \ldots x^{(n)}_{i_n}x^{(n)}_{j_n}.
\end{align}
If $k=2$, the restriction to $x_0$'s means that we get
\[
\iota_{11(0^{n-2})}\vert_{x_0}=d_{11(0^{n-2})} \left(x_0^{(3)}\ldots x_0^{(n)}\right)^2.
\]
If $k>2$, at the next transvectant step we set the $x$'s in
(\ref{step1}) to $y$'s, multiply by the fundamental form $f$ and apply
$\Omega_3$ to get
\[
\Omega_3 \left[ \sum_{k_1 \ldots k_n} a_{k_1 \ldots k_n}x^{(1)}_{k_1} \dots x^{(n)}_{k_n} \right] \left[\sum_{i_3 \ldots i_n; j_1 \ldots j_n} \left(a_{11i_3 \ldots i_n}a_{00j_3 \ldots j_n}-a_{10i_3 \ldots i_n}a_{01j_3 \ldots j_n} \right) y^{(3)}_{i_3}y^{(3)}_{j_3} \ldots y^{(n)}_{i_n}y^{(n)}_{j_n}\right] \Bigm \vert_{y \to x}.
\]
If we are restricted to $x_0$'s, we must have $k_1=k_2=0$ since no
further $\Omega$ operations are applied in these index positions and
so these $x$'s will be unchanged. Only certain sets of indices are
consistent with a $y_0$ remaining after applying $\Omega_3$ to
$x^{(3)}_{k_{3}}y^{(3)}_{i_{3}}y^{(3)}_{j_{3}}$; namely
(1) $k_{3}=0$, $i_{3}=1$, $j_{3}=0$; (2) $k_{3}=0$,
$i_{3}=0$, $j_{3}=1$; (3) $k_{3}=1$, $i_{3}=0$,
$j_{3}=0$. The result of this operation is of the form
\[
\sum_{i,jk} (\alpha_{i,j,k}+\beta_{i,j,k})x_0^{(1)}x_0^{(2)}x_0^{(3)}x^{(4)}_{i_4}x^{(4)}_{j_4}x^{(4)}_{k_4} \ldots x^{(n)}_{i_{n}}x^{(n)}_{j_{n}}x^{(n)}_{k_{n}},
\]
where $\alpha_{i,jk}$, $\beta_{i,j,k}$ are terms in the $a$'s with
compound indices $i=\{i_4 \ldots i_n\}$, etc., and $\alpha_{i,jk}$
comes from the conditions (1) and (2) above on index sets:
\[
\alpha_{i,j,k}=a_{000k_4 \ldots k_n} \cR_{3,1}\left[a_{110 i_4 \ldots i_n}a_{000 j_4 \ldots j_n}-a_{010i_4 \ldots i_n}a_{100j_4 \ldots j_n}\right],
\]
whereas from condition (3) we get
\[
\beta_{i,j,k}=a_{001k_4 \ldots k_4}\left[a_{110i_4 \ldots i_n}a_{000j_4 \ldots j_n}-a_{010i_4 \ldots i_n}a_{100j_4 \ldots j_n}\right].
\]
If $k=3$ this simplifies to 
\begin{align}
\iota_{111(0^{n-3})}\vert_{x_0}=\left[ a_{000(0^{n-3})}\cR_{3,1}d_{110 (0^{n-3})} -2a_{001(0^{n-3})}d_{110(0^{n-3})}  \right]\left(x_0^{(1)}x_0^{(2)}x_0^{(3)}\right) \left(x_0^{(4)} \ldots x_0^{(n)}\right)^3.
\end{align}
Using (\ref{partitions}), a straightforward calculation shows that 
\begin{align}\label{induction3}
d_{111 (0^{n-3})}=a_{000(0^{n-3})}\cR_{3,1}d_{110 (0^{n-3})} -2a_{001 (0^{n-3})}d_{110 (0^{n-3})}.
\end{align}
Repeating the above argument, we have
\begin{align}\label{iota-k}
\iota_{1^k0^{n-k}}\vert_{x_0}=d_{1^k0^{n-k}} \left(x_0^{(1)} \ldots x_0^{(k)}\right)^{k-2} \left(x_0^{(k+1)} \ldots x_0^{(n)}\right)^{k},
\end{align}
and the generalisation of (\ref{induction3}) is
\begin{align}\label{induction-k}
d_{1^k0^{n-k}}=a_{0^n}\cR_{k,1}d_{1^{k-1}0^{n-k+1}} -2a_{0^{k-1}10^{n-k+1}}d_{1^{k-1}0^{n-k+1}}.
\end{align}
This last equation has a straightforward interpretation. When
evaluating $f(\psi)$ by the Taylor series (\ref{taylor}), the
coefficient of, say, $e_1$ can be obtained by differentiating
$\frac{\partial}{\partial e_1}f(\psi)$ and setting $e_i=0$, for all
$i$. Writing $\psi=a+r$, where $a=a_{0 \ldots 0}$, we find
\[
\frac{\partial}{\partial e_1}f(a+r)\vert_{e_i=0}=f^\prime(a+r)\frac{\partial}{\partial e_1}r\vert_{e_i=0}=f^\prime(a)a_{10\ldots 0}=f^\prime(a)\cR_{1,1}a.
\]
We can interpret the last expression above as the formal derivative of
$f(\psi)$ using the raising operator $\cR_{1,1}$, and similarly the
coefficient of any product $e_{i_1} \ldots e_{i_q}$ is the result of
formal derivatives by $\cR_{i_1,1} \ldots \cR_{i_q,1}$. This can
indeed be taken as the {\em definition} of the expansion of $f(\psi)$,
as in \cite{AbergMitchison09}. For the $\log$ function, the
coefficient of $e_1 \ldots e_{k-1}$ is $c_{1^{k-1}0^{n-k+1}}$, and
the coefficient of $e_1 \ldots e_k$, namely $c_{1^k0^{n-k}}$ is
obtained by applying $\cR_{k,1}$ to
$c_{1^{k-1}0^{n-k+1}}$. Differentiating $\log(\psi)$ and using
$c_{1^q0^{n-q}}=d_{1^q0^{n-q}}(a_{0^n})^{-q}$ gives (\ref{induction-k}).

From (\ref{iota-k}) and the definition of the inner product
(\ref{inner-product}) we find that 
\[
\langle \left(\iota_{1^k0^{n-k}}\vert_{x_0}\right)|
\left(\iota_{1^k0^{n-k}}\vert_{x_0} \right)\rangle=\xi|d_{1^k0^{n-k}}|^2,
\]
where $\xi$ is the constant given in the Proposition. With the
restriction to $x_0$'s we therefore get, up to the factor $\xi$, the
term in the formula for $I_{1^k0^{n-k}}$ (Theorem \ref{explicit})
where $k_p=0$ for all $p$. To complete the proof, one observes that,
allowing $k$ $x_1^{(i)}$'s introduces $k$ 1's into the $a$'s at position
$i$, and is equivalent to applying $\cR_{i_p,k_p}$ to
$d_{1^k0^{n-k}}$. The values of the coefficients $\alpha^{i_p}_{k_p}$
are given by the derivative inner product.

\end{proof}
This enables us to recognise some of the four-qubit invariants in
\cite{Luque07}. Up to a constant factor, we have the following
identifications:
\begin{align}
\label{1000} I_{1000} &\leftrightarrow A_{1111},\\
\label{1100} I_{1100} &\leftrightarrow \langle B_{0022}|B_{0022}\rangle,\\
\label{1110} I_{1110} &\leftrightarrow \langle C_{1113}|C_{1113}\rangle,
\end{align}
which, with permutations of indices in (\ref{1100}) and (\ref{1110}),
yields 11 corresponding pairs. Note that we use different letters for
the invariants, and our subscripts indicate the total number of 1's at
a given position in successive transvection operations; see Table
\ref{transvectants}.

We now come to the third way of defining the lift. Suppose that a
covariant $p_{l_1 \ldots l_n}$ is derived by some sequence of
transvectant operations. Define its $i$th lift $p_{l_1 \ldots 0_i
  \ldots l_n}$ by adding an index position in the $i$th position in
the ground form, and applying the same transvectant operations, but
with an `0' added to the transvectant indices in the $i$th position.
\begin{proposition}\label{cov}
  If $P_{l_1 \ldots l_n}=\langle p_{l_1 \ldots l_n} | p_{l_1 \ldots
    l_n} \rangle$ is the invariant derived from the covariant $p_{l_1
    \ldots l_n}$, then the $i$th lift of $P_{l_1 \ldots l_n}$ is given
  by $P_{l_1 \ldots 0_i \ldots l_n}=\langle p_{l_1 \ldots 0_i \ldots
    l_n} | p_{l_1 \ldots 0_i \ldots l_n} \rangle$.
\end{proposition}
\begin{proof}
  Because there is a 0 at position $i$ in the transvectant indices,
  $\Omega_i$ is never applied during the transvection operations. This
  means that we get all possible products of $x_0^{(i)}$ and
  $x_1^{(i)}$, and the terms with $k$ $x_1^{(i)}$'s correspond to
  products of $a$'s with $k$ 1's in index position 1.
\end{proof}

As an example, consider the invariant of highest degree for 3-qubits
in \cite{Luque07}. It is given (in our notation) by $H_{222}=\langle
h_{222}|h_{222} \rangle$, where
\begin{align}\label{h}
h_{222}=(f,(f,(f,f)^{110})^{001})^{111}.
\end{align}
Equivalently, $H_{222}=|$Det$(\ket{\psi})|^2$, where
$\mbox{Det}(\ket{\psi})$ is the
hyperdeterminant (\ref{Det}). From (\ref{h}), the
lift of $h_{222}$ at position 4 is
$h_{2220}=(f,(f,(f,f)^{1100})^{0010})^{1110}$, and therefore by
Proposition \ref{cov} $H_{2220}=\langle h_{2220}|h_{2220} \rangle$ is
the lift at position 4 of $H_{222}$. In the terminology of
\cite{Luque07}, $H_{222}$ is $D_{000}$ and $H_{2220}$ is $\langle
D_{0004}|D_{0004} \rangle$.

Another example is obtained by putting $g_{1111}=(f,f)^{1111}$ and
setting $G_{1111}=\langle g_{1111}|g_{1111} \rangle$. This 4th degree
invariant can be written
\[
G_{1111}=a_{0000}a_{1111}-\left(a_{1000}a_{0111}+\mbox{ permutations }\right) +\left(a_{1100}a_{0011}+\mbox{ permutations }\right).
\]
In general, for each $k$, we add a new $2k$-party invariant
$G_{1^{2k}}=\langle g_{1^{2k}}|g_{1^{2k}} \rangle$, where
$g_{1^{2k}}=(f,f)^{1^{2k}}$. Together with all its lifts, the $G$
family comprises $\binom{n}{2}+\binom{n}{4}+\binom{n}{6} +\ldots
=2^{n-1}-1$ independent invariants of degree 4 for an $n$-qubit
system.  This coincides with the family $B_d$ in \cite{Luque07}.

\section{Conclusions}

We have seen that local unitary invariants come in families, related
by a tracing operation (Proposition \ref{tracedef}). An invariant for
$n$-qubit states can be 'lifted' to give invariants of $n+1$-qubit
states; when this process is repeated, one gets an exponentially large
family of invariants.

One important family is derived from twirled cumulants. These can be
shown to be algebraically independent within each $n$-qubit system,
and asymptotically the total number of such invariants is half the
dimension of the orbit space, $D(n)$. For three and four qubits, these
invariants are in fact already known \cite{Luque07} -- see (\ref{100})
to (\ref{111}), and (\ref{1000}) to (\ref{1110}) -- but their
connection to cumulants seems not to be recognised, nor their
relationships to each other through lifting operations. Another
example is the hyperdeterminant family: for $n=3$, $D_{000}$ in the
notation in \cite{Luque07} is the 3-tangle \cite{Coffman00}, and
$\langle D_{0004}|D_{0004} \rangle$ and its permutations are lifted
3-tangles.

Many of the invariants are closely related to separability of
states. The hyperdeterminant, in its guise as the 3-tangle, is an
entanglement measure for mixed states. The vanishing of members of the
cumulant family can be used to characterise multipartite separability
of pure states (see Theorem \ref{separability-criterion}). We can also
ask which states maximise these invariants. For instance, $I_{11}$
attains its maximum for a Bell state, and $I_{111}$ for a
GHZ. $I_{110}$ is maximised by $\Psi \otimes \ket{0}$, with $\Psi$ a
Bell state.  We can regard this as an example of monogamy of
entanglement \cite{monogamy,Coffman00}, with $I_{110}$ detecting
entanglement between the first two systems, and achieving its maximum
when they are unentangled with the third system. Similarly, one might
conjecture that $I_{1^k0}$ is maximised by states of the form
$\ket{\mu} \otimes \ket{0}$, where the $k$-qubit state $\ket{\mu}$
maximises $I_{1^k}$.

How far can the ideas here can be generalised beyond pure qubit
states? The cumulant-based invariants can be applied to mixed states
via the map (\ref{JtohatJ}). However, Theorem \ref{separability} tells
us only about correlation rather than mixed-state separability. The
results also fail to generalise for pure states where the local
dimension exceeds two. We can construct invariants, and Theorem
\ref{separability} holds, but the invariants are not algebraically
independent. This is seen even for two qutrits, where we have four
members of the cumulant family, namely $I_{11}$, $I_{12}$, $I_{21}$
and $I_{22}$, whereas there are only two independent invariants
\cite{Gu09}. Since four polynomial equations is the correct number to
characterise separability, the simple relationship between invariants
and separability cannot hold for $d>2$. Nonetheless the basic concept
of lifts and families still applies in all these wider contexts.

\section{Acknowledgements}

I thank Johan {\AA}berg and Tony Sudbery for comments on the emerging
manuscript, Graeme Segal for help with formulating the algebra of
multi-partite states, and Markus Grassl for many helpful comments and
for pointing out some egregious blunders in my exposition of invariant
theory in the previous version of this paper.

\section{APPENDIX}

\subsection{An alternative cumulant-based invariant}\label{zhou-section}

There is a very different way of relating cumulants and invariants,
due to Zhou et al. \cite{Zhouetal06}. Given an $n$-party mixed state
$\rho$, one defines its cumulant by analogy with (\ref{partitions}) as
\begin{align}\label{zhou}
\rho_c=\sum_\pi (-1)^{|\pi|-1}(|\pi|-1)! \bigotimes_{i=1}^{|\pi|}\rho_{\pi_i},
\end{align}
where $\rho_{\pi_i}$ is the result of tracing out from $\rho$ all
systems apart from those with labels in $\pi_i$. For instance, for three
systems
\[
\rho_c=\rho-(\rho_1\otimes \rho_{23}+\rho_2 \otimes \rho_{13}+\rho_3
\otimes \rho_{12})+2\rho_1 \otimes \rho_2\otimes \rho_3.
\]
The cumulant operator given by (\ref{zhou}) is not in general a state,
but Zhou et al. propose $M(\rho)=\frac{1}{2}\tr|\rho_c|$ as a measure
of correlation of the mixed state $\rho$. It is manifestly invariant
under local unitaries, and, because of the general property cumulants
have of vanishing on products, $M(\rho)=0$ whenever $\rho=\rho_S
\otimes \rho_T$. For pure states, this means it vanishes when states
are separable.

It therefore seems to have formal similarities to our cumulant-based
invariants, and one can carry this further by defining, in line with
Proposition \ref{tracedef}, the lift of $M$ to be $M(\tr
\proj{\psi})$. We can in fact adopt parallel notation to the $I$'s,
writing, for a 3-qubit state for example,
$M_{111}(\ket{\psi})=\frac{1}{2}\tr|\left(\proj{\psi}\right)_c|$,
$M_{110}(\ket{\psi})=\frac{1}{2}\tr|\left(\tr_3\proj{\psi}\right)_c|$,
and so on. Then $M_{i_1 \ldots i_n}(\ket{\psi})=0$ for any
$\pi$-separable $\ket{\psi}$ where $\{i_1 \ldots i_n\}$ splits
$\pi$. Furthermore, for 3-qubit states $M_{111}(\ket{\psi})=0$ is
sufficient for separability of $\ket{\psi}$ (\cite{Zhouetal06},
Theorem 3), and the same is true if $I_{111}(\ket{\psi})=0$.

These similarities prompt the question of whether there is a
functional connection. Can one write $M_{i_1 \ldots
  i_n}(\ket{\psi})=F(I_{i_1 \ldots i_n}(\ket{\psi}))$, for some
function $F$? For 2-qubit states,
$M_{11}=I_{11}+\sqrt{I_{11}}$. However, there is only one 2-qubit
invariant for normalised states, so a functional relationship here is
unsurprising. For 3-qubit states of the form
$\ket{\Psi}=a\ket{000}+b\ket{111}$ one finds
\begin{align}\label{f1}
  M_{111}=6I_{111}\sqrt{1-4I_{111}}+2\sqrt{I_{111}+I_{111}^2-4I_{111}^3},
\end{align}
whereas, for states of the form
$\ket{\phi}=a\ket{100}+b\ket{010}+c\ket{001}$
\begin{align}\label{f2}
(M_{111}-\frac{I_{111}}{2})^3-\frac{1}{4}I_{111}=0.
\end{align}
Since (\ref{f1}) and (\ref{f2}) do not define the same function of
$I_{111}$, $M_{111}$ must depend on other invariants besides
$I_{111}$.

\end{document}